\documentclass[conference]{IEEEtran}
\IEEEoverridecommandlockouts

\usepackage{amsthm}

\theoremstyle{plain}
\newtheorem{theorem}{Theorem}

\theoremstyle{definition}
\newtheorem{definition}[theorem]{Definition}

\newcommand{\set}{S}

\newcommand{\dcp}{\set_N} 
\newcommand{\sizedcp}{N} 
\newcommand{\sdcp}{n} 

\newcommand{\ip}{\set_M} 
\newcommand{\sizeip}{M}
\newcommand{\sip}{m} 

\newcommand{\ns}{\mathcal{D}} 

\usepackage{latexsym}
\usepackage{booktabs}
\usepackage{enumitem}

\usepackage{cite}
\usepackage{amsmath,amssymb,amsfonts}
\usepackage{algorithmic}
\usepackage{graphicx}
\usepackage{textcomp}
\usepackage{xcolor}

\def\BibTeX{{\rm B\kern-.05em{\sc i\kern-.025em b}\kern-.08em
    T\kern-.1667em\lower.7ex\hbox{E}\kern-.125emX}}
\begin{document}

\title{Enabling Privacy-preserving Model Evaluation in Federated Learning via Fully Homomorphic Encryption\\
}
\author{

\IEEEauthorblockN{Cem Ata Baykara}
\author{
    \IEEEauthorblockN{Cem Ata Baykara\IEEEauthorrefmark{1}\IEEEauthorrefmark{2}\IEEEauthorrefmark{3}, 
    Ali Burak Ünal\IEEEauthorrefmark{1}\IEEEauthorrefmark{2}, 
    Mete Akgün\IEEEauthorrefmark{1}\IEEEauthorrefmark{2}}
    \\
    \IEEEauthorblockA{\IEEEauthorrefmark{1}Medical Data Privacy and Privacy-Preserving Machine Learning, 
    University of Tübingen, Tübingen, 72076, Germany}
    \\
    \IEEEauthorblockA{\IEEEauthorrefmark{2}Institute for Bioinformatics and Medical Informatics, 
    University of Tübingen, Tübingen, 72076, Germany}
    \\
    \IEEEauthorblockA{\IEEEauthorrefmark{3}Corresponding author(s). E-mail(s): cem.baykara@uni-tuebingen.de}
    
}
}

\maketitle

\begin{abstract}
Federated learning has become increasingly widespread due to its ability to train models collaboratively without centralizing sensitive data. While most research on FL emphasizes privacy-preserving techniques during training, the evaluation phase also presents significant privacy risks that have not been adequately addressed in the literature. In particular, the state-of-the-art solution for computing the area under the curve (AUC) in FL systems employs differential privacy, which not only fails to protect against a malicious aggregator but also suffers from severe performance degradation on smaller datasets. Moreover, although the individual data label privacy is protected, this DP-based approach reveals side information such as noisy statistics, as well as individual client's and global AUC scores, making the approach not applicable for more sensitive real-life applications.

To overcome these limitations, we propose a novel evaluation method that leverages fully homomorphic encryption. To the best of our knowledge, this is the first work to apply FHE to privacy-preserving model evaluation in federated learning while providing verifiable security guarantees. In our approach, clients encrypt their true-positive and false-positive counts based on predefined thresholds and submit them to an aggregator, which then performs homomorphic operations to compute the global AUC without ever seeing intermediate or final results in plaintext. We offer two variants of our protocol: one secure against a semi-honest aggregator and one that additionally detects and prevents manipulations by a malicious aggregator. Besides providing verifiable security guarantees, our solution achieves superior accuracy across datasets of any size and distribution, eliminating the performance issues faced by the existing state-of-the-art method on small datasets and its runtime is negligibly small and independent of the test-set size. Experimental results confirm that our method can compute the AUC among 100 parties in under two seconds with near-perfect (99.93\%) accuracy while preserving complete data privacy.
\end{abstract}


\section{Introduction}

Today, machine learning has become a powerful tool applicable to countless fields. However, with increasing concerns about data privacy and confidentiality, governments worldwide are focusing more on regulations such as the GDPR\footnote{General Data Protection Regulation, European Union} and CCPA\footnote{California Consumer Privacy Act}, which regulate how personal data should be managed. As a result, privacy-preserving machine learning techniques like Federated Learning (FL) have received significant research interest \cite{li2020federated, MLSYS2019_bd686fd6, MOTHUKURI2021619}.

FL is defined as a setting where multiple input parties,  collectively train a machine learning model, often through a central server (e.g. service provider, aggregator), while keeping their local data private \cite{kairouz2021advances}. Although the term FL was first proposed with cross-device applications in mind \cite{mcmahan2017communication, brendan2017blog}, the significance of utilizing FL in other fields has vastly increased, which includes less number of more reliable participants (e.g. collaborating organizations, hospitals) \cite{kairouz2021advances}. Data partition between the input parties plays an important role in the operation of the FL system. Based on how the data is partitioned across participants, FL is defined in two categories: \textit{Horizontal} and \textit{Vertical} FL. In this paper, our focus is on \textit{Horizontal} FL, where data is split by entity (e.g. a patient), and the data entities owned by each party are disjoint from those owned by other parties.

Current research on the privacy of FL systems mostly focuses on the training phase, where the model architecture or parameters are communicated or shared between the input parties or the aggregator \cite{zhao2020idlg, zhu2019deep, wei2020federated}. However, in practice, after the training phase is completed, the input parties also need to test the performance of the model on test data. Input parties with a limited number of test samples cannot accurately assess the performance of the collaboratively trained global model. This problem can be solved through a cooperative solution where the test samples of all input parties are used for performance evaluation. However, as with the training samples, sharing test samples would cause privacy problems \cite{stoddard2014differentially, matthews2013examination}. Therefore, a privacy-preserving solution for collaborative model performance evaluation in FL systems is required.

The state-of-the-art method for privacy-preserving model performance evaluation in FL systems utilizes Differential Privacy (DP) \cite{sun2023dpauc} to compute the area under the Receiver Operating Characteristics curve. While differential privacy can protect the privacy of individual samples, it fails to prevent a semi-honest aggregator from inferring side information about the local data, such as local sample sizes, class distributions of input parties, and both local and global performance of the global model. Additionally, this method provides no security against a malicious aggregator capable of tampering with the data during computation. More importantly, our experimental results, presented in Section \ref{experiments}, demonstrate that as the size of the global test dataset decreases, the AUC values computed by this method become entirely unreliable and unusable.

To address these issues, we propose a novel approach that leverages Fully Homomorphic Encryption (FHE) to preserve privacy during model performance evaluation. Our method efficiently, accurately, and securely computes the AUC in \textit{Horizontal} FL systems while ensuring security against both semi-honest and malicious adversaries. To the best of our knowledge, our method is the first to utilize FHE for privacy-preserving model evaluation in \textit{Horizontal} FL systems while providing security against both semi-honest and malicious adversaries. While our focus is on AUC computation due to its complexity in distributed settings, we note that our method can be extended to compute other performance metrics such as accuracy, precision, and recall, with even lower computational overhead. Our contributions are as follows:

\begin{itemize}
\item We empirically demonstrate that the existing DP-based state-of-the-art method suffers from severe limitations, particularly in small-data scenarios where the added noise leads to unreliable results.

\item To the best of our knowledge, we propose the first method to apply FHE for privacy-preserving model evaluation in \textit{Horizontal} FL systems while ensuring verifiable security guarantees. Our approach provides complete data privacy and protects against both semi-honest and malicious aggregators, addressing a critical security gap in the literature.

\item Our experimental results show that our method significantly outperforms the state-of-the-art method and achieves near-perfect AUC computation accuracy regardless of data size and distribution across clients.

\item We demonstrate that our method remains highly scalable in terms of both computation time and communication cost. Despite utilizing FHE, our approach achieves AUC computation within practical time limits, making it feasible for real-world FL deployments. Notably, our method's efficiency is independent of the test set size, ensuring scalability across various application settings.
\end{itemize}


\section{Preliminaries}

In this section, we present some background information on the concepts used throughout the paper.

\subsection{ROC curve and AUC}

The ROC curve plots False Positive Rate (FPR) on the x-axis and True Positive Rate (TPR) on the y-axis. TPR for a given threshold value $\theta$ is defined as $\textit{TPR}^\theta=\frac{TP^\theta}{TP^\theta+FN^\theta}$, where True Positive is denoted as $TP$ and False Negative as $FN$. Similarly, FPR for a given threshold value $\theta$ is defined as $\textit{FPR}^\theta=\frac{FP^\theta}{FP^\theta+TN^\theta}$, where False Positive is denoted as $FP$ and True Negative as $TN$.

The AUC is obtained by measuring the area between the ROC curve and the x-axis. Calculating the AUC depends on the specific problem and available information. For a binary classification problem with $\ns$ test samples, the AUC of the ROC curve can be computed as follows:
\begin{equation}
    AUC= {\textstyle\sum_{i=1}^{\ns}} \frac{(TPR^{i} + TPR^{i-1}) \cdot (FPR^{i} - FPR^{i-1})}{2}
\label{eq:exactAUC}
\end{equation}

If we observe Equation \ref{eq:exactAUC}, we can see that the formula makes use of the TPR and FPR values to compute the AUC. However, using either of these equations poses a challenge if we would like to compute the AUC through FHE since it is difficult to perform division in FHE schemes. In order to minimize the number of division operations, Equation \ref{eq:exactAUC} can be rewritten by replacing $TPR^i$ with $\frac{TP^{i}}{TP^{\ns}}$ and $FPR^i$ with $\frac{FP^{i}}{FP^{\ns}}$ as follows:
\begin{equation}
    AUC=\frac{{\textstyle\sum_{i=1}^{\ns}} (TP^{i} + TP^{i-1}) \cdot (FP^{i} - FP^{i-1})}{2 \cdot TP^{\ns} \cdot FP^{\ns}}
\label{eq:onedivauc}
\end{equation}

Notice that only a single division operation is enough to compute the AUC. We make use of this property to compute the AUC while utilizing FHE.

\subsection{Trapezoidal Rule}

The trapezoidal rule approximates definite integrals or areas under curves \cite{yeh2002using}. It divides the curve into sub-intervals, forming trapezoids under the curve, and calculates their areas, summing them for an approximate area under the curve. The formula for approximating the area under the ROC curve resembles Equation \ref{eq:onedivauc}. Instead of using the threshold samples inside the test data, we use a set of decision points to approximate the AUC. If $\dcp$ denotes the set of $\sizedcp$ decision points partitioning $[0,1]$, the approximation formula is as follows:

\begin{equation}
AUC \approx {\textstyle\sum_{n=1}^{\sizedcp-1}} \frac{(TP^{\sdcp} + TP^{\sdcp-1}) \cdot (FP^{\sdcp} - FP^{\sdcp-1})}{2 \cdot TP^{\sizedcp-1} \cdot FP^{\sizedcp-1}}
\label{eq:traprule}
\end{equation}

where, similar to Equation \ref{eq:onedivauc}, the $TP^{\sizedcp-1}$ and $FP^{\sizedcp-1}$ values in the denominator represent the total number of positive and negative samples, respectively.

\subsection{Homomorphic Encryption}

Homomorphic encryption (HE) enables operations on encrypted data, producing results mirroring those of operations on plaintexts. Additive HE schemes, like the Paillier scheme \cite{paillier1999public}, support addition over encrypted data, while schemes such as RSA \cite{rivest1978method} and ElGamal \cite{elgamal1985public} facilitate multiplication, falling under multiplicative HE schemes. Partial HE schemes support only one operation, either addition or multiplication. FHE schemes, first introduced by Gentry \cite{gentry2009fully} using lattice theory, support both operations. FHE encompasses various schemes categorized by supported data types, with noise introduced during encryption and key generation.

\section{Related Work}

Most research on FL systems is focused on privacy preservation during the training phase \cite{zhao2020idlg, zhu2019deep, wei2020federated, sav2022privacy} with very few methods proposed for privacy preservation during the model evaluation phase \cite{unal2023ppaurora, sun2023dpauc}.

\subsection{Multi-Party Computation}

A multi-party computation based approach for calculating the exact AUC within a system having multiple data sources is proposed by \cite{unal2023ppaurora}. Although the proposed method is capable of computing the ground truth AUC, the use of privacy-preserving merge sort results in infeasible computation times. The experimental results presented by \cite{unal2023ppaurora} show that the proposed method takes approximately 1400 seconds to compute the AUC on just 8000 global data samples with 8 parties. Moreover, since there is no mechanism to prevent collusion between entities during data outsourcing, the proposed method does not provide security against malicious adversaries capable of tampering with the computation. In contrast, our approach can compute the approximate AUC with near-perfect accuracy while providing security against both semi-honest and malicious adversaries in just a few seconds.

\subsection{Differential Privacy}

The state-of-the-art method, DPAUC$_{\text{Lap}}$ \cite{sun2023dpauc}, utilizes differential privacy for AUC computation in FL systems. DPAUC$_{\text{Lap}}$, which we will refer to as DPAUC throughout the rest of this paper, employs Laplacian noise to achieve differential privacy and uses the Trapezoidal Rule to approximate the AUC value. In DPAUC, upon receiving the global model, participants in the FL system compute their TP, FP, TN, and FN statistics based on a set of decision points $\dcp$ with a fixed size. Then, each client adds Laplacian noise with a total privacy budget of $\epsilon \in \{1,2,4,8\}$ to their statistics. The noise introduced by the input parties scales with the number of decision points, where $\epsilon = \sizedcp * (\epsilon_{TP} + \epsilon_{FP} + \epsilon_{TN} + \epsilon_{FN})$. Sun et al. \cite{sun2023dpauc} report that they set $\epsilon_{TP} = \epsilon_{FP} = \epsilon_{TN} = \epsilon_{FN} = \epsilon'$, resulting in a total privacy budget of $4 \sizedcp \epsilon'$. These noisy statistics are then transmitted to the aggregator to compute the global AUC using the Trapezoidal approximation shown in Eq \ref{eq:traprule} \cite{sun2023dpauc}.

Because the amount of noise being added is independent of the data size, this leads to completely unreliable AUC approximations when the global data size is small. More specifically, when the number of test samples is low, the signal-to-noise ratio (SNR) becomes too low, as the magnitude of the added Laplacian noise remains constant while the useful signal (i.e., the true statistical counts) decreases. This limitation, not addressed in \cite{sun2023dpauc}, was mitigated in their experiments by using a large test set of 458,407 samples. However, our experiments in Section \ref{experiments} show that DPAUC yields unusable AUC estimates in small-sample settings, a common scenario in fields like medicine \cite{kononenko2001machine, foster2014machine}, a key application area for FL. 



\section{Methods}

This section first outlines our threat model and then details our approaches for computing AUC in an FL system in the presence of both the semi-honest aggregator and the malicious aggregator.

\subsection{Notations}

For the rest of the paper, we represent the number of input parties as $\ip$ whose size is $\sizeip$. We denote the set of decision points as $\dcp$ whose size is represented as $\sizedcp$. $TP^\sdcp_\sip$ and $FP^\sdcp_\sip$ represent the number of true positives and false positives of input party $Party_\sip$ on decision point $\sdcp$ where \(n \in N\), respectively. An exception being $\sdcp = \sizedcp - 1$ in the malicious setting, where $T^\sdcp_\sip$ denotes $(TP^\sdcp_\sip + TP^{\sdcp-1}_\sip)$ and $F^\sdcp_\sip$ denotes $(FP^\sdcp_\sip - FP^{\sdcp-1}_\sip)$. Moreover, in malicious setting, these values are masked by some random values. $X_\sip$ represents the dataset of the input party $Party_\sip$.

\subsection{Threat Model}

Our proposed AUC computation methods address two distinct threat models: semi-honest and malicious aggregators. Additionally, we assume non-collusion between the aggregator and input parties to ensure the security of the private key used in FHE.

In the semi-honest model, all participants, including the aggregator and input parties, follow the protocol but may attempt to infer additional information from the data they process. To mitigate this risk, our semi-honest protocol uses FHE, where each input party encrypts their data before sending it to the aggregator. The aggregator then performs computations on the encrypted data without learning anything about the underlying plaintext.

In the malicious setting, the aggregator may deviate from the protocol, attempting to manipulate computations while still trying to extract information. To overcome this, we introduce additional verification steps. Input parties randomize their values and perform computations twice, ensuring that any manipulation by the aggregator can be detected. This extension provides robust security even against a malicious aggregator as proven through our security analysis.

\subsection{Semi-honest Setting} \label{sh_fhauc}

In this section, we introduce our solution for computing AUC in a privacy-preserving way when dealing with semi-honest input parties and a semi-honest aggregator.

\begin{figure*}[ht]
\centering
\includegraphics[width=0.9\textwidth]{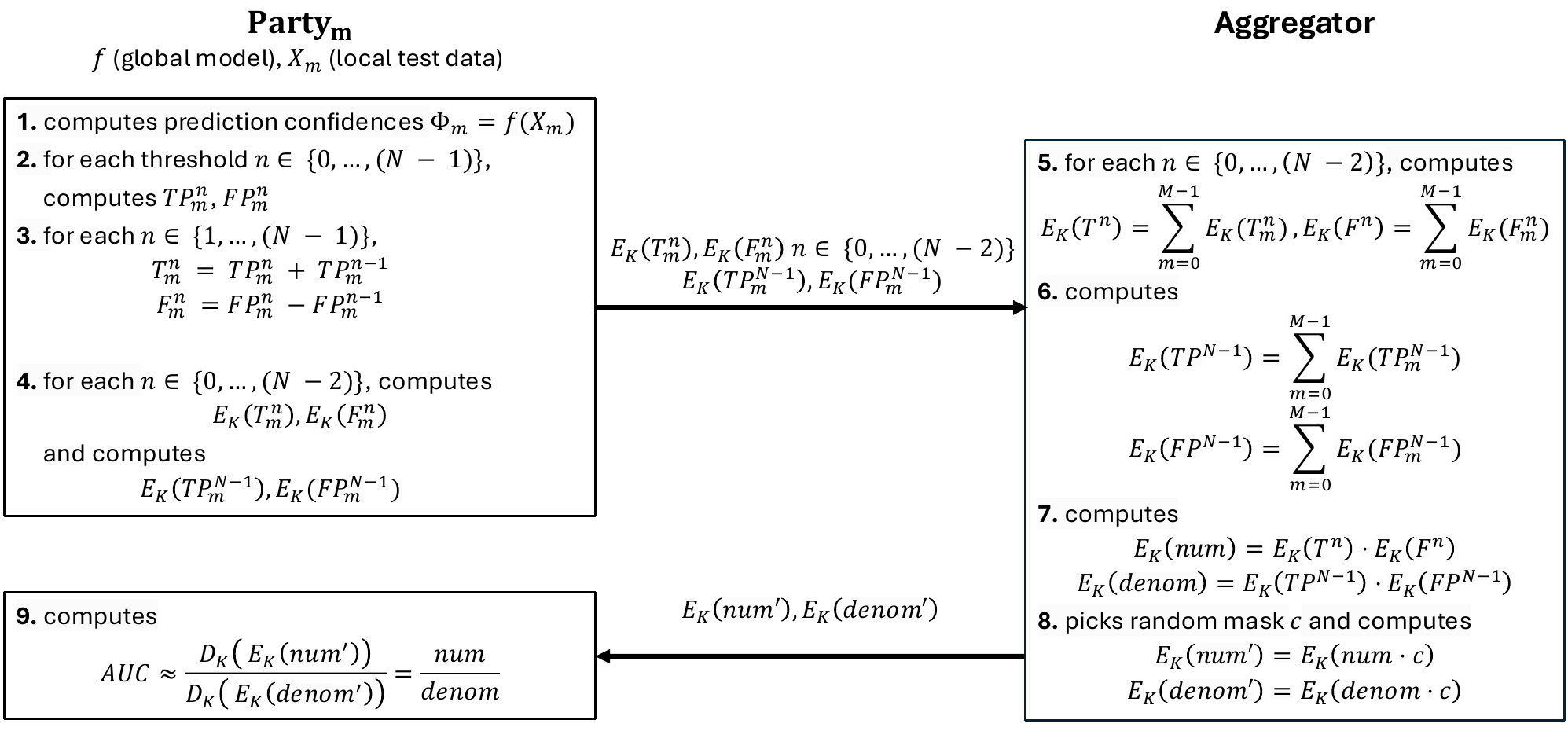} 
\caption{The workflow of the semi-honest setting.}
\label{fig:FHAUC}
\vspace{0.3cm}
\end{figure*}

\subsubsection{Setup}

The generation and distribution of a public-private key pair for FHE are managed through a public key infrastructure. The process begins with the aggregator randomly selecting an input party within the FL system to generate the key pair required for AUC computation. The selected party shares the public key with the aggregator, which is necessary for performing FHE operations on the ciphertexts, and sends the private key to other input parties by encrypting it with their public keys. The public keys of the input parties can be obtained through a certificate authority. The use of a certificate authority for the distribution of the FHE private key is a widely adopted practice in various solutions that employ FHE \cite{liu2019secure, zhang2020batchcrypt}. Furthermore, internet security mainly relies on the use of certificate authorities. While our primary focus lies in ensuring the privacy and security of the computation phase, it is important to note that the security of the key distribution phase is not within the scope of our paper. The issue of secure key distribution within a distributed environment often stands as an independent challenge, which is typically addressed by existing secure distribution protocols. This challenge also falls beyond the scope of the existing FHE-based solutions outlined in the literature \cite{liu2019secure, zhang2020batchcrypt, aono2017privacy}.

\subsubsection{Secure Computation}

After the completion of the setup phase, both the input parties and the aggregator proceed with the AUC computation. The workflow of the semi-honest setting is shown in Figure \ref{fig:FHAUC}. The algorithm consists of nine steps:

\begin{enumerate}
    \item Each input party calculates their prediction confidence scores $\Phi_\sip=f(X_\sip)$ by using the global model $f$ on their respective local test data $X_\sip$.
    
    \item For each decision point $\sdcp \in \{0,\ldots,(\sizedcp-1)\}$, each input party $Party_\sip$ computes the TP$^\sdcp_\sip$, and FP$^\sdcp_\sip$ statistics based on their local data. After computing TP$^\sdcp_\sip$ and FP$^\sdcp_\sip$, each input party $Party_\sip$ also obtains TP$^{\sizedcp-1}_\sip$ and FP$^{\sizedcp-1}_\sip$, which denote the total number of positively labeled samples and the total number of negatively labeled samples within $X_m$, respectively.
    
    \item Each input party $Party_\sip \in \{0,\ldots,(\sizeip-1)\}$ computes $T^{\sdcp}_\sip = TP^\sdcp_\sip + TP^{\sdcp-1}_\sip$ and $F^{\sdcp}_\sip = FP^\sdcp_\sip - FP^{\sdcp-1}_\sip$, where $\sdcp \in \{1,\ldots,(\sizedcp-1)\}$. Here, $TP^{\sdcp-1}_\sip$ and $FP^{\sdcp-1}_\sip$ denote the $TP^\sdcp_\sip$ and $TP^\sdcp_\sip$ vectors that have been shifted one position to the left, respectively.
    
    \item  Each input party $Party_\sip \in \{0,\ldots,(\sizeip-1)\}$ encrypts their statistics using the FHE private key as follows: $E_{K}(\text{T}^{\sdcp}_\sip)$ and $E_{K}(\text{F}^{\sdcp}_\sip)$ for each $\sdcp \in \{1,\ldots,(\sizedcp-1)\}$, and $E_K(\text{TP}^{\sizedcp-1}_\sip)$ and $E_K(\text{FP}^{\sizedcp-1}_\sip)$. These encrypted statistics are then sent to the aggregator.
    
    \item The aggregator utilizes FHE to aggregate all the encrypted $\text{T}^{\sdcp}_\sip$ and $\text{F}^{\sdcp}_\sip$ of input parties to obtain $E_{K}(\text{T}^{\sdcp})$ and $E_{K}(\text{F}^{\sdcp})$ for each $\sdcp \in \{0,\ldots,(\sizedcp-2)\}$.
    
    \item The aggregator utilizes FHE to aggregate all the encrypted $\text{TP}^{\sizedcp-1}_\sip$ and $\text{FP}^{\sizedcp-1}_\sip$ of input parties to obtain $E_{K}(\text{TP}^{\sizedcp-1})$ and $E_{K}(\text{FP}^{\sizedcp-1})$.
    
    \item The aggregator utilizes FHE to compute:
    \begin{align*}
  E_{K}(num) &= E_{K}({\textstyle\sum_{\sdcp=1}^{\sizedcp-1}}\text{T}^{\sdcp - 1} \cdot \text{F}^{\sdcp - 1}) \\
  E_{K}(denom) &= E_{K}(2 \cdot\text{TP}^{\sizedcp-1} \cdot \text{FP}^{\sizedcp-1})
    \end{align*}
    
    \item Aggregator generates a cryptographically secure random floating-point value $c$, and uses it to compute $E_K(num \cdot c)$ and $E_K(denom \cdot c)$. It then sends these computed ciphertexts to the input parties.
    
    \item Each input party decrypts the encrypted ciphertexts to obtain 
    \begin{align*}
        D_{K}(E_K(num \cdot c)) &= num' \\
        D_{K}(E_{K}(denom \cdot c)) &= denom'
    \end{align*}
    They then compute
    \begin{align*}
        \frac{num'}{denom'} =\frac{num \cdot c}{denom \cdot c}= \frac{num}{denom} \approx AUC
    \end{align*}
    resulting in an approximation of the exact global AUC value.
\end{enumerate}

\subsection{Malicious Setting}
\label{mal_fhauc}

This section details our solution for privacy-preserving AUC computation when the input parties are semi-honest and the aggregator is malicious.

In this scenario, each input party randomizes all $T$ and $F$ values, encrypts these randomized values under their private key, and sends the encrypted data to the aggregator. Both the random generation stage in the setup phase and the entire secure computation phase are each executed twice between the input parties and the aggregator. In the first run, the aggregator computes 
\[
  (r_0 \cdot \textit{num} + r_1 \cdot \textit{denom},\, r_2 \cdot \textit{denom}),
\]
and in the second run, it computes 
\[
  (r_0' \cdot \textit{num} + r_1' \cdot \textit{denom},\, r_2' \cdot \textit{denom}).
\]
Here, $r_0,r_1,r_2,r_0',r_1',r_2'$ are random values known only to the input parties. Ultimately, each input party recovers the AUC from 
\[
  \frac{r_0 \cdot \textit{num} + r_1 \cdot \textit{denom}}{r_2 \cdot \textit{denom}}
\]
and the second AUC (denoted AUC$'$) from
\[
  \frac{r_0' \cdot \textit{num} + r_1' \cdot \textit{denom}}{r_2' \cdot \textit{denom}}.
\]
If these two results match, it confirms the aggregator followed the protocol correctly; any deviation by the aggregator would cause the equality check to fail.

\subsubsection{Setup}

The setup phase for the malicious setting begins the same way as in the semi-honest setting. The input parties generate common random values \(r_3,\, r_4,\, r_5,\, r_6,\, r_7,\, r_8\), and two arrays, $tr^{n}_\sip$ and $fr^{n}_\sip$, each of size $(\sizedcp+1)\,\sizeip\,S$ containing random values. We set $r_0 = r_3 r_4$, $r_1 = r_5 r_6$, $r_2 = r_7 r_8$, where $n \in \{0,\ldots,\sizedcp\}$, $\sip \in \{0,\ldots,\sizeip-1\}$, and $S$ is the split count. The input parties also generate a common random permutation $\pi$ of size $S \cdot \sizedcp$ and an array of random bits $b$ of size $(\sizedcp+1)$.

\subsubsection{Secure Computation}

The secure computation steps are carried out as follows. Steps 1--8 are performed by the input parties, steps 9--11 by the aggregator, and step 12 again by the input parties.

\textbf{1.} Each input party $\sip$ computes its prediction confidence scores \(\Phi_\sip = f(X_\sip)\), using the global model $f$ on its local data $X_\sip$.

\textbf{2.} For each threshold $\sdcp \in \{0,\dots,\sizedcp - 1\}$ and each party $\sip$, compute local TP$^\sdcp_\sip$ and FP$^\sdcp_\sip$ based on $X_\sip$.

\textbf{3.} Each party $\sip \in \{0,\ldots,\sizeip - 2\}$ computes
\begin{align*}
    T^{\sdcp-1}_\sip = r_3 \bigl(\text{TP}^\sdcp_\sip + \text{TP}^{\sdcp-1}_\sip + tr^{\sdcp-1}_\sip\bigr), \\
  F^{\sdcp-1}_\sip = r_4 \bigl(\text{FP}^\sdcp_\sip - \text{FP}^{\sdcp-1}_\sip + fr^{\sdcp-1}_\sip\bigr),
\end{align*}

while the party $\sizeip - 1$ computes
\begin{align*}
  T^{\sdcp-1}_{\sizeip-1} 
    = r_3 \Bigl(\text{TP}^\sdcp_{\sizeip-1} + \text{TP}^{\sdcp-1}_{\sizeip-1} 
    - {\textstyle\sum_{j=0}^{\sizeip-2}} tr^{\sdcp-1}_j\Bigr), \\
  F^{\sdcp-1}_{\sizeip-1} 
    = r_4 \Bigl(\text{FP}^\sdcp_{\sizeip-1} - \text{FP}^{\sdcp-1}_{\sizeip-1} 
    - {\textstyle\sum_{j=0}^{\sizeip-2}} fr^{\sdcp-1}_j\Bigr),
\end{align*}
for $n \in \{1,\ldots,\sizedcp - 1\}$.

\textbf{4.} Similarly, each party $\sip \in \{0,\ldots,\sizeip - 2\}$ computes
\begin{align*}
  T^{\sizedcp-1}_\sip = r_5 \bigl(\text{TP}^{\sizedcp-1}_\sip + tr^{\sizedcp-1}_\sip\bigr), \\
  F^{\sizedcp-1}_\sip = r_6 \bigl(\text{FP}^{\sizedcp-1}_\sip + fr^{\sizedcp-1}_\sip\bigr),
\end{align*}
and the party $\sizeip-1$ computes
\begin{align*}
    T^{\sizedcp-1}_{\sizeip-1} 
        = r_5 \Bigl(\text{TP}^{\sizedcp-1}_{\sizeip-1} 
        - {\textstyle\sum_{j=0}^{\sizeip-2}} tr^{\sizedcp-1}_j\Bigr), \\
      F^{\sizedcp-1}_{\sizeip-1} 
        = r_6 \Bigl(\text{FP}^{\sizedcp-1}_{\sizeip-1} 
        - {\textstyle\sum_{j=0}^{\sizeip-2}} fr^{\sizedcp-1}_j\Bigr).
\end{align*}
  
\textbf{5.} Each party $\sip \in \{0,\ldots,\sizeip-2\}$ computes
\[
  DT_\sip = r_7 \bigl(\text{TP}^{\sizedcp-1}_\sip + tr^{\sizedcp}_\sip\bigr),
  \quad
  DF_\sip = r_8 \bigl(\text{FP}^{\sizedcp-1}_\sip + fr^{\sizedcp}_\sip\bigr),
\]
and the party $\sizeip-1$ computes
\begin{align*}
  DT_{\sizeip-1} 
    = r_7 \Bigl(\text{TP}^{\sizedcp-1}_{\sizeip-1} 
    - {\textstyle\sum_{j=0}^{\sizeip-2}} tr^{\sizedcp}_j\Bigr), \\
  DF_{\sizeip-1} 
    = r_8 \Bigl(\text{FP}^{\sizedcp-1}_{\sizeip-1} 
    - {\textstyle\sum_{j=0}^{\sizeip-2}} fr^{\sizedcp}_j\Bigr).
\end{align*}

\textbf{6.} Each party $\sip$ splits $T^{\sdcp}_{\sip}$ or $F^{\sdcp}_{\sip}$ into $S$ additive shares according to a random bit $b[\sdcp]$. If $b[\sdcp]$ is $0$, then 
\[
  T^{\sdcp}_{\sip} = {\textstyle\sum_{s=0}^{S-1}} T^{\sdcp}_{\sip}[s],
  \quad
  \text{and}
  \quad
  F^{\sdcp}_{\sip}[s] = F^{\sdcp}_{\sip}\;\,\forall\,s.
\]
If $b[\sdcp]$ is $1$, the roles of $T^{\sdcp}_{\sip}$ and $F^{\sdcp}_{\sip}$ are reversed.

\textbf{7.} Each party concatenates all $T^{\sdcp}_{\sip}[s]$ into a vector $T^{\text{all}}_{\sip}$ and all $F^{\sdcp}_{\sip}[s]$ into a vector $F^{\text{all}}_{\sip}$, each of size $S \cdot \sizedcp$:
\[
  T^{\text{all}}_{\sip} 
    = [T^{0}_{\sip}[0],\dots,T^{0}_{\sip}[S-1], \dots, T^{\sizedcp-1}_{\sip}[S-1]],
\]
\[
  F^{\text{all}}_{\sip} 
    = [F^{0}_{\sip}[0],\dots,F^{0}_{\sip}[S-1], \dots, F^{\sizedcp-1}_{\sip}[S-1]].
\]
Each party then permutes these vectors using the common permutation $\pi$.

\textbf{8.} Each party $\sip$ encrypts all elements of $T^{\text{all}}_{\sip}$, $F^{\text{all}}_{\sip}$, $DT_\sip$, and $DF_\sip$ with the FHE private key and sends the resulting ciphertexts to the aggregator.

\textbf{9.} For each $i \in \{0,\dots,S \cdot \sizedcp - 1\}$, the aggregator homomorphically sums the ciphertexts across parties:
\begin{align*}
    E_{K}\bigl(T^{\text{all}}[i]\bigr) 
    = E_{K}\Bigl({\textstyle\sum_{\sip=0}^{\sizeip-1}} T^{\text{all}}_{\sip}[i]\Bigr),\\
    E_{K}\bigl(F^{\text{all}}[i]\bigr) 
    = E_{K}\Bigl({\textstyle\sum_{\sip=0}^{\sizeip-1}} F^{\text{all}}_{\sip}[i]\Bigr), \\
      E_{K}(\text{DT}) 
    = E_{K}\Bigl({\textstyle\sum_{\sip=0}^{\sizeip-1}} DT_\sip\Bigr),\\
  E_{K}(\text{DF}) 
    = E_{K}\Bigl({\textstyle\sum_{\sip=0}^{\sizeip-1}} DF_\sip\Bigr).
\end{align*}

\textbf{10.} Using FHE, the aggregator computes:
\begin{align*}
E_K\Bigl({\textstyle\sum_{i=0}^{S\sizedcp-1}} T^{\text{all}}[i]\,F^{\text{all}}[i]\Bigr)
  &= E_K\bigl(r_3r_4\,\textit{num} + r_5r_6\,\textit{denom}\bigr) \\
  &= E_K\bigl(r_0\,\textit{num} + r_1\,\textit{denom}\bigr).
\end{align*}

\[
  E_{K}\bigl(\text{DT}\cdot\text{DF}\bigr)
  \;=\;E_{K}\bigl(r_7\,r_8\,\textit{denom}\bigr)
  \;=\;E_{K}\bigl(r_2 \cdot \textit{denom}\bigr).
\]

\textbf{11.} The aggregator generates a random mask $c$, then computes:
\[
  E_{K}\bigl(c \cdot (r_0\,\textit{num} + r_1\,\textit{denom})\bigr),
  \quad
  E_{K}\bigl(c \cdot r_2 \cdot \textit{denom}\bigr),
\]
and sends these ciphertexts to the input parties.

\textbf{12.} Each input party decrypts to obtain 
\[
  c \cdot (r_0\,\textit{num} + r_1\,\textit{denom}) 
  \quad\text{and}\quad 
  c \cdot r_2 \cdot \textit{denom}.
\]

\subsubsection{Verification}

The random generation phase (setup) and the secure computation phase each run twice (using the same inputs). The input parties then receive two outputs and perform the following check:
\begin{align*}
    \text{AUC} 
    = \biggl(\frac{c \cdot (r_0 \cdot \textit{num} 
      + r_1 \cdot \textit{denom})}
      {c \cdot r_2 \cdot \textit{denom}} 
      - \frac{r_1}{r_2}\biggr) 
      \cdot \frac{r_2}{r_0}
    = \frac{\textit{num}}{\textit{denom}} \\
    \text{AUC}'
    = \biggl(\frac{c' \cdot (r_0' \cdot \textit{num}
      + r_1' \cdot \textit{denom})}
      {c' \cdot r_2' \cdot \textit{denom}}
      - \frac{r_1'}{r_2'}\biggr)
      \cdot \frac{r_2'}{r_0'}
    = \frac{\textit{num}}{\textit{denom}}
\end{align*}

If $\text{AUC} = \text{AUC}'$, the aggregator is deemed to have followed the protocol honestly; any deviation would cause a mismatch in this equality check.

\section{Outsourced Division}

Performing division in FHE is challenging, as all current FHE schemes, regardless of generation, natively support only addition and multiplication but not division. In our method, we explored two numerical approximation techniques to allow the aggregator to perform the division between \textit{num} and \textit{denom} homomorphically and directly share the resulting AUC score with the clients. However, both methods presented significant challenges.

\subsection{Chebyshev Approximation}

This approach requires predefined upper and lower bounds to construct an accurate polynomial approximation of the reciprocal function \cite{rivlin1974chebyshev}. However, in our setting, depending on the number of local test samples held by each client, the \textit{denom} value of the AUC can range anywhere from \(10^3\) to \(10^{10}\), making it infeasible to pick a single interval that guarantees accuracy. We experimented with having all clients scale their \(TP\) and \(FP\) vectors prior to encryption, effectively scaling the resulting \textit{num} and \textit{denom} into a predefined range so the aggregator could accurately perform the approximation and compute the AUC score directly. However, this approach requires information sharing between clients, such as the maximum and minimum number of local test samples held by participants in the system.

\subsection{Newton-Raphson Approximation}

This approach is iterative and requires an initial guess close to the reciprocal for efficient convergence \cite{cetin2015arithmetic}. Similar to the Chebyshev approximation, without prior knowledge of the range of the \textit{denom}, selecting an accurate starting point is difficult. Without this knowledge, the number of iterations required for an accurate approximation becomes excessive, leading to an impractical number of ciphertext-to-ciphertext multiplications. This significantly increases computational cost and reduces efficiency, making it unsuitable for real-time federated learning applications. Scaling is again possible but requires information sharing between the parties.

\subsection{Our Approach}

To overcome these challenges, our proposed method outsources the division step to the clients through a masking approach. By applying a multiplicative random mask to hide the actual values of \textit{num} and \textit{denom}, our method eliminates the need for information sharing between the parties or with the aggregator. It also requires no prior knowledge of the denominator’s range while maintaining computational efficiency. In our security analysis (see Appendix A), we prove that our approach of masking and outsourcing the division step to the clients reveals no more information to the clients than directly outputting the AUC score.


\section{Security Analysis}
A detailed security analysis of our protocol is provided in Appendix. The analysis considers both semi-honest and malicious aggregators and demonstrates that the proposed method ensures input confidentiality and output correctness under standard security assumptions. We refer the reader to the appendix for formal definitions, threat models, and detailed proofs.


\section{Experiments} \label{experiments}

\begin{table*}[t]
\centering
\caption{Performance analysis of our method compared to DPAUC on Criteo dataset with $\epsilon=\{1,8\}$ and 15 parties.}\label{table:results_criteo}

\begin{tabular*}{\textwidth}{@{\extracolsep{\fill}} cccc ccc}
\toprule
Data Size & TensorFlow & scikit-learn & $\sizedcp$ & Ours & DPAUC \cite{sun2023dpauc} ($\epsilon=1$) & DPAUC \cite{sun2023dpauc} ($\epsilon=8$) \\
\midrule
                         &            &            & $25$  & $0.764580 \pm 0.0$ & $112.8856 \pm 1029.426$ & $40.56802 \pm 72.81583$ \\
   $10^{2}$              & $0.764792$ & $0.765417$ & $50$  & $0.761465 \pm 0.0$ & $66.11053 \pm 3486.300$ & $-6.93354 \pm 1280.965$ \\
                         &            &            & $100$ & $0.764411 \pm 0.0$ & $192.3464 \pm 4973.634$ & $29.72084 \pm 4835.073$ \\
    \midrule
                         &            &            & $25$  & $0.752502 \pm 0.0$ & $69.507704 \pm 725.146$ & $0.786719 \pm 0.212634$ \\
   $10^{3}$              & $0.752482$ & $0.752500$ & $50$  & $0.752444 \pm 0.0$ & $-8.473372 \pm 1237.06$ & $2.572270 \pm 19.06697$ \\
                         &            &            & $100$ & $0.752375 \pm 0.0$ & $42.184663 \pm 1757.05$ & $42.63417 \pm 62.44075$ \\
    \midrule
                         &            &            & $25$  & $0.745805 \pm 0.0$ & $0.792406 \pm 0.236237$ & $0.746144 \pm 0.028154$ \\
   $10^{4}$              & $0.746190$ & $0.746184$ & $50$  & $0.746052 \pm 0.0$ & $1.032923 \pm 0.953103$ & $0.754784 \pm 0.049641$ \\
                         &            &            & $100$ & $0.746266 \pm 0.0$ & $-8.722478\pm 113.148$  & $0.741087 \pm 0.107208$ \\
    \midrule
                         &            &            & $25$  & $0.744820 \pm 0.0$ & $0.744391 \pm 0.004721$ & $0.744792 \pm 0.000549$ \\
   $4.58 \times 10^{5}$  & $0.745416$ & $0.745416$ & $50$  & $0.745283 \pm 0.0$ & $0.745462 \pm 0.010124$ & $0.745351 \pm 0.001053$ \\
                         &            &            & $100$ & $0.745380 \pm 0.0$ & $0.749182 \pm 0.018712$ & $0.745672 \pm 0.001883$ \\
    
    \bottomrule
\end{tabular*}

\end{table*}

\begin{table*}[t]
\centering
\caption{Performance analysis of our method compared to DPAUC on CelebA dataset with $\epsilon=\{1,8\}$ and 15 parties.}
\label{table:gender}

\begin{tabular*}{\textwidth}{@{\extracolsep{\fill}} cccc ccc}
\toprule
Data Size & TensorFlow & scikit-learn & $\sizedcp$ & Ours & DPAUC \cite{sun2023dpauc} ($\epsilon=1$) & DPAUC \cite{sun2023dpauc} ($\epsilon=8$) \\
\midrule
          &            &              & 25   & $0.893400 \pm 0.0$       & $63.00177 \pm 723.6120$  & $-0.97248 \pm 28.83271$ \\
$10^{2}$  & $0.891400$ & $0.891600$   & 50   & $0.888800 \pm 0.0$       & $87.28492 \pm 787.6992$  & $-43.2448 \pm 496.1687$ \\
          &            &              & 100  & $0.892600 \pm 0.0$       & $244.4928 \pm 2981.016$  & $-93.4737 \pm 657.0762$ \\
\midrule
          &            &              & 25   & $0.937070 \pm 0.0$       & $8.817562 \pm 408.6866$  & $0.991520 \pm 0.320545$ \\
$10^{3}$  & $0.938614$ & $0.938588$   & 50   & $0.938250 \pm 0.0$       & $127.9842 \pm 1369.205$  & $2.261357 \pm 9.195085$ \\
          &            &              & 100  & $0.938700 \pm 0.0$       & $253.1984 \pm 2616.424$  & $8.875182 \pm 50.45507$ \\
\midrule
          &            &              & 25   & $0.909860 \pm 0.0$       & $1.009118 \pm 0.249137$  & $0.904852 \pm 0.027761$ \\
$10^{4}$  & $0.911410$ & $0.911406$   & 50   & $0.910900 \pm 0.0$       & $1.387757 \pm 1.807835$  & $0.914889 \pm 0.057772$ \\
          &            &              & 100  & $0.911290 \pm 0.0$       & $10.76247 \pm 102.0806$  & $0.894566 \pm 0.115002$ \\
\midrule
          &            &              & 25   & $0.950080 \pm 0.0$       & $0.947500 \pm 0.022812$  & $0.950197 \pm 0.003212$ \\
$10^{5}$  & $0.954020$ & $0.954052$   & 50   & $0.952370 \pm 0.0$       & $0.950787 \pm 0.049092$  & $0.953180 \pm 0.006660$ \\
          &            &              & 100  & $0.953422 \pm 0.0$       & $0.980881 \pm 0.103526$  & $0.953366 \pm 0.011731$ \\
\bottomrule
\end{tabular*}

\end{table*}

We implemented our methods using OpenFHE \cite{al2022openfhe}, an open-source Fully Homomorphic Encryption library that supports the CKKS scheme \cite{cheon2017homomorphic}, allowing approximate computations over real or complex number vectors. In all our experiments, we use a scaling factor length of 50 bits and set the security level to \texttt{HEStd\_128\_classic}, providing at least a 128-bit security level against the aggregator. All the experiments were conducted on a MacBook Pro with an Apple M2 Pro chip, utilizing a single core and 32 GB of memory.

\textbf{Datasets:} We test the performance of our method and how it compares against DPAUC on two different datasets. First, we used the Criteo\footnote{https://www.kaggle.com/c/criteo-display-ad-challenge/data} dataset, consistent with Sun et al.'s evaluation of DPAUC. Criteo is a large-scale industrial binary classification dataset with approximately 45 million samples, consisting of 26 categorical and 13 real-valued features. We followed the data pre-processing steps outlined by Sun et al. to maintain consistency \cite{sun2023dpauc}. Mirroring Sun et al.'s approach, our training data comprised $90\%$ of the entire Criteo dataset, with the remaining $10\%$ used for test data. As for our second dataset, we used celebA, an image classification dataset\footnote{https://www.kaggle.com/datasets/ashishjangra27/gender-recognition-200k-images-celeba}. This dataset comprises over 200,000 images of celebrities for gender prediction tasks. Each image is of size $178 \times 218$ pixels and is in color. We chose an image classification dataset to maximize the number of threshold samples, thereby allowing for a clearer observation of our proposed method's performance. To assess the performance of both our method and DPAUC across various global data sizes, we adjusted the train-to-test data ratios while keeping 10,000 validation samples constant.

\textbf{Model:} For the Criteo data, we trained a modified DLRM (Deep Learning Recommendation Model) architecture presented in \cite{naumov2019deep} for 3 epochs. As for the gender classification data, we employed a basic CNN architecture comprising 13 layers and trained the model for 4 epochs for each train-test split ratio. It is crucial to emphasize that our primary objective is to showcase the performance of our proposed AUC computation method and how it compares to the state-of-the-art, and thus, the performances of the models are not our primary concern.

\textbf{Ground-truth AUC:} To evaluate the approximation performance, we utilize two AUC computation libraries: TensorFlow\footnote{https://www.tensorflow.org} and scikit-learn\footnote{https://scikit-learn.org}. We present their computed results as the ground-truth and compare them with our method and DPAUC. Consistent with the approach of Sun et al. \cite{sun2023dpauc}, in our experiments, we set the number of thresholds to 1,000 for TensorFlow and utilized the default values for other parameters.

\textbf{Comparison Metric:} For DPAUC, following Sun et al. \cite{sun2023dpauc}, we run the same setting for 100 times and report the corresponding mean and standard deviation of the computed AUC as our comparison metric.

\subsection{Performance Analysis}

We compare the approximation performance of our proposed method and DPAUC against the ground truth AUC score computed using the libraries described above. To implement DPAUC, we utilized diffprivlib \cite{diffprivlib}, a general-purpose differential privacy library for Python developed by IBM. In all our experiments, we set $\epsilon$ to $1$ and $8$ for DPAUC, representing the smallest and largest privacy budgets used by \cite{sun2023dpauc} to evaluate the effectiveness of their method while providing the most and least privacy, respectively. We employed 15 participants in the FL system, and each data sample was uniformly distributed at random among the participants. As our method operates deterministically, the number of participants in the FL system and the data distribution among them do not influence the resulting AUC. The only factors influencing the resulting AUC of our approach are the number of decision points used by the participants and the total number of global data samples.

Tables \ref{table:results_criteo} and \ref{table:gender} present the performance of our method and DPAUC against the ground truth with varying numbers of decision points and different amounts of global data. Because the approximation performance of our method in both the semi-honest and the malicious setting is the same, we did not report it twice for convenience. In Tables \ref{table:results_criteo} and \ref{table:gender}, we can see that our method consistently outperforms DPAUC in terms of approximation performance. Moreover, increasing the number of decision points leads to improved approximation performance for our approach, but it reduces the performance of DPAUC. This is because in DPAUC, the noise introduced by the input parties scales with the number of decision points utilized by the input parties. Since the amount of noise being added in DPAUC is independent of the data size, this leads to completely unusable approximations when the global data size is small, as clearly shown in Tables \ref{table:results_criteo} and \ref{table:gender}. In contrast, the AUC scores computed by our method are robust and independent of the global data size. It is important to note that the performance of DPAUC is influenced by both IID and Non-IID data distributions, as discussed in \cite{sun2023dpauc}, while the performance of our method remains unaffected by the distribution of data among input parties.

\subsection{Computation Time Analysis}

Finally, we examined the computation time required for the aggregator to compute the AUC in both the malicious setting and the semi-honest setting. The computation times are analyzed based on the number of participants and the number of decision points used, as depicted in Figures \ref{fig:FHAUC_a} and \ref{fig:FHAUC_b}. To assess the impact of the number of decision points on the computation time, we maintain a constant number of 15 participants within the FL system. Similarly, the influence of the number of parties on the computation time is evaluated while keeping the number of decision points $\sizedcp=100$ constant. As for the malicious setting, we set the number of splits to $S=4$.

\begin{figure}[ht]
\centering
\includegraphics[width=8cm]{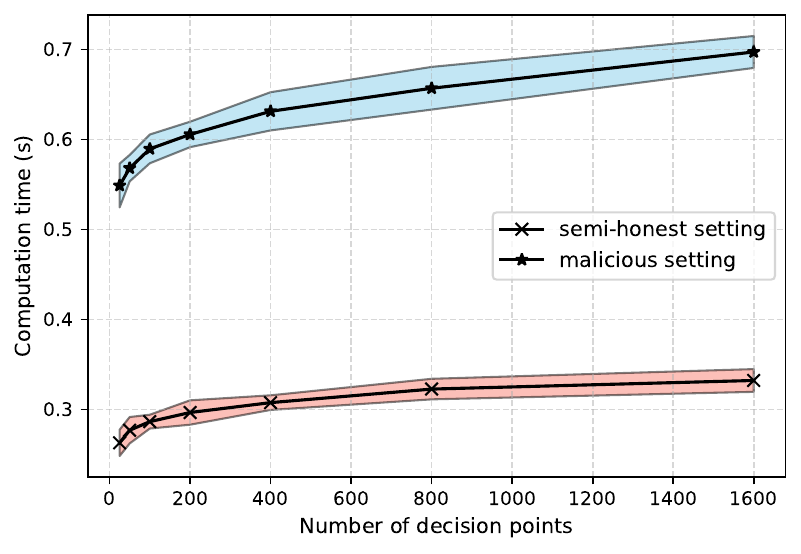}
\caption{The AUC computation time evaluated against the number of decision points for a fixed number of 15 parties.}
\label{fig:FHAUC_a}
\end{figure}

\begin{figure}[ht]
\centering
\includegraphics[width=8cm]{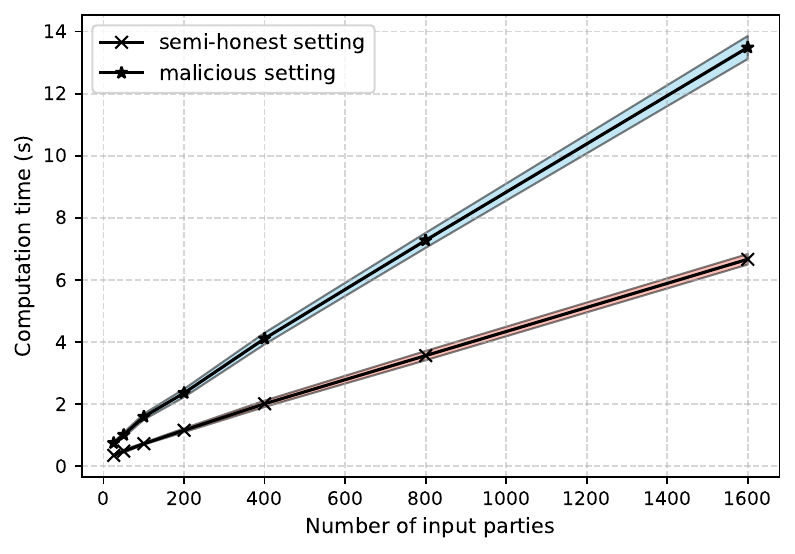}
\caption{The AUC computation time evaluated against the number of parties for a fixed number of decision points ($\sizedcp=100$).}
\label{fig:FHAUC_b}
\end{figure}

Because our approach utilizes SIMD (Single Instruction Multiple Data) support provided by OpenFHE \cite{al2022openfhe}, the AUC computation is highly efficient despite employing FHE. As an example, our method is capable of computing the AUC of an FL system with 100 parties with $99.93\%$ accuracy in just $0.68$ seconds in the semi-honest setting and $1.64$ seconds in the malicious setting, regardless of data size. As depicted in Figure \ref{fig:FHAUC_a}, we observe a logarithmic relationship between the number of decision points used and the computation time. Upon examining Figure \ref{fig:FHAUC_b}, we observe a linear relationship between the number of parties and the computation time. Comparing the semi-honest setting against the malicious setting, we see that the computation time of the malicious setting is almost double for all settings. This is to be expected because the AUC computation is performed twice in the malicious setting. Although the computation times of our method are well within the range for real-time use, it is important to emphasize that our analysis is not intended as a comparison with the DPAUC method. Since DPAUC performs all computations on plain-text, its computation time is inherently much lower. Instead, our goal is to demonstrate that despite utilizing FHE, our approach remains highly efficient and well-suited for practical applications.

\subsection{Communication Cost Analysis}

We evaluated the communication cost of our proposed method under both semi-honest and malicious security settings, from both the client and server sides. Table~\ref{table:comm_cost} presents the communication cost for varying numbers of clients.

\begin{table}[t]
\centering
\caption{Communication cost analysis of our method under semi-honest and malicious security settings for different number of clients.}
\begin{tabular}{crrr}
\toprule
\textbf{Setting} & \textbf{\# Clients} & \textbf{Client Side} & \textbf{Server Side} \\
\midrule
            & 5  & 6.81 MB  & 34.05 MB  \\
Semi-Honest & 10 & 6.81 MB  & 68.10 MB  \\
            & 15 & 6.81 MB  & 102.15 MB \\
\midrule
            & 5  & 13.62 MB & 68.10 MB  \\
Malicious   & 10 & 13.62 MB & 136.20 MB \\
            & 15 & 13.62 MB & 204.30 MB \\
\bottomrule
\end{tabular}

\label{table:comm_cost}
\end{table}





Our method's communication cost is independent of the global sample size due to the approximation technique used. Instead, it primarily depends on the ring dimension, which we set to \(2^{14}\), supporting up to \(2^{13} = 8192\) slots. As a result, in the semi-honest setting, for decision points \(N < 8192\), the communication cost remains constant for the client side, regardless of the local dataset size. Since our experiments show that \(N = 100\) already provides highly accurate AUC approximations, input parties do not need to choose larger \(N\). Hence, the communication cost for the server side is only affected by the number of input parties inside the FL system.

In the malicious setting, the communication cost is simply doubled compared to the semi-honest setting, as the AUC computation is performed twice. Specifically, as shown in Table \ref{table:comm_cost}, there is no additional increase beyond this doubling. This is because our design of random \(S\) splits does not introduce extra ciphertexts. Given that our implementation supports up to \(2^{13} = 8192\) slots, as long as \(S N < 2^{13}\), the communication cost remains limited to twice that of the semi-honest setting.

\section{Conclusion}

We proposed a novel privacy-preserving evaluation method for FL systems that eliminates the limitations of the existing state-of-the-art approach based on differential privacy. Unlike the previous method, our approach ensures complete data privacy while providing security against both semi-honest and malicious aggregators. By leveraging FHE, we enable secure and tamper-proof computation of the AUC without revealing any intermediate or final results in plaintext.

Our experimental results demonstrate that our method significantly outperforms the DP-based alternative in terms of both accuracy and reliability, particularly in scenarios with limited test data. We further show that, despite utilizing FHE, our approach remains computationally efficient, achieving AUC computation among 100 parties in under two seconds regardless of data size. Additionally, our security analysis proves that the protocol preserves privacy even in the presence of a malicious aggregator attempting to tamper with the computation.

Although we focus on AUC due to its computational and privacy challenges, our approach can be readily adapted to compute other performance metrics, many of which involve simpler operations and would benefit even more from FHE.

To the best of our knowledge, this is the first work to apply FHE to privacy-preserving model evaluation in federated learning while providing verifiable security guarantees. Our findings highlight the feasibility of FHE-based approaches for secure federated model evaluation and open new directions for research into efficient and scalable privacy-preserving performance metrics.

\bibliography{sample-base}

\newpage

\appendix
\label{appendix:security_analysis}

We analyze the security of our method for both the semi-honest and the malicious setting. Our security analysis is based on the following assumptions:

\begin{enumerate}
    \item The aggregator cannot decrypt the shared ciphertexts without access to the private key.
    \item The communication between the individual participants and the aggregator is secure.
    \item The distribution of the FHE private key is performed securely.
    \item There is no collusion between the aggregator and the input parties.
\end{enumerate}

\subsection{Security of the Multiplicative Masking}\label{sec:masking_proof}

We begin by proving that our approach of outsourcing the division operation to the clients through a random multiplicative mask is secure.

\begin{definition}[AUC]
Let $\textit{num}, \textit{denom}\in\mathbb{N}$ be two integers (with $\textit{denom}\neq 0$) such that the AUC is defined as
\[
\text{AUC} = \frac{\textit{num}}{\textit{denom}}.
\]
\end{definition}

\begin{definition}[Blinded Output]
Let $c\in\mathbb{R}^+$ be a random floating point number drawn independently from a distribution $\mathcal{L}$ (with sufficiently wide support). The blinded outputs are defined as
\[
\textit{num}' = \textit{num}\cdot c,\quad \textit{denom}' = \textit{denom}\cdot c.
\]
\end{definition}

\begin{theorem}
Assume that $\textit{num}$ and $\textit{denom}$ are drawn from some domain and that the only information leaked in the unblinded approach is the ratio 
\[
\frac{\textit{num}}{\textit{denom}},
\]
i.e., the AUC. Then, the secure computation that outputs the blinded pair 
\[
\left(\textit{num}',\textit{denom}'\right)=\left(\textit{num}\cdot c,\ \textit{denom}\cdot c\right)
\]
leaks no more information than outputting the AUC. In particular, for every candidate pair $(\textit{num},\textit{denom})$ that yields a given AUC, there exists a $c$ such that the blinded pair corresponds to that candidate. Consequently, the adversary cannot infer additional information about $\textit{num}$ and $\textit{denom}$ beyond the AUC.
\end{theorem}

\begin{proof}
Let $\mathcal{A}$ be an adversary whose goal is to learn additional information about $\textit{num}$ and $\textit{denom}$ beyond the AUC. In the unmasked approach, the secure computation outputs 
\[
v = \frac{\textit{num}}{\textit{denom}},
\]
so the adversary only learns the ratio $v$ and may generate many candidate pairs $(\textit{num}, \textit{denom})$ such that 
\[
\frac{\textit{num}}{\textit{denom}} = v.
\]

In the blinded approach (our method), the secure computation outputs 
\[
(\textit{num}', \textit{denom}') = (\textit{num}\cdot c,\ \textit{denom}\cdot c).
\]
Note that the ratio remains the same:
\[
\frac{\textit{num}'}{\textit{denom}'} = \frac{\textit{num}\cdot c}{\textit{denom}\cdot c} = \frac{\textit{num}}{\textit{denom}} = v.
\]
Thus, the adversary learns the same AUC as in the first case.

We now show that the additional values $\textit{num}'$ and $\textit{denom}'$ do not allow the adversary to distinguish between different candidate pairs that yield the same AUC. For any candidate pair $(\textit{num}_0,\textit{denom}_0)$ satisfying 
\[
\frac{\textit{num}_0}{\textit{denom}_0} = v,
\]
suppose that the actual pair is $(\textit{num},\textit{denom})$ and that the secure computation chooses $c$ uniformly at random from a distribution $\mathcal{L}$ over an interval with wide support. Then, the output of our method is
\[
(\textit{num}',\textit{denom}') = (\textit{num}\cdot c,\ \textit{denom}\cdot c).
\]
Given the blinded output, the adversary can compute $v$ but, for any candidate $(\textit{num}_0,\textit{denom}_0)$ with $v = \frac{\textit{num}_0}{\textit{denom}_0}$, there exists a unique $c_0$ such that 
\[
\textit{num}' = \textit{num}_0 \cdot c_0 \quad \text{and} \quad \textit{denom}' = \textit{denom}_0 \cdot c_0.
\]
That is, 
\[
c_0 = \frac{\textit{num}'}{\textit{num}_0} = \frac{\textit{denom}'}{\textit{denom}_0}.
\]
Since $c$ is drawn independently of $\textit{num}$ and $\textit{denom}$, the adversary cannot prefer one candidate pair over another; for any candidate pair that yields the same AUC, a corresponding scaling factor $c_0$ exists that maps it to the observed blinded output.

Thus, the mapping 
\[
f: (\textit{num},\textit{denom}, c) \mapsto (\textit{num}\cdot c, \textit{denom}\cdot c)
\]
is a bijection between the set of pairs $(\textit{num},\textit{denom})$ (that satisfy $v=\frac{\textit{num}}{\textit{denom}}$) and the set of possible blinded outputs (parameterized by the random choice of $c$). Consequently, the adversary's view in the second scenario is statistically equivalent to that in the first scenario (outputting only $v$), aside from the additional randomness introduced by $c$, which is independent of $\textit{num}$ and $\textit{denom}$.

Therefore, the secure computation that outputs the blinded pair $(\textit{num}', \textit{denom}')$ does not leak any additional information beyond the AUC.
\end{proof}

\subsubsection{Practical Security Considerations}

Our method employs multiplicative masking to obscure the individual values of \textit{num} and \textit{denom} during the outsourced division operation. Specifically, we release the blinded pair \((\textit{num}',\;\textit{denom}')=(\textit{num} \cdot c,\;\textit{denom} \cdot c)\) for a randomly drawn positive scalar \(c\). The primary goal of this masking is not to achieve formal $\lambda$-bit security (as defined in cryptographic masking literature), but rather to prevent the inference of the exact values of num and denom while still allowing the correct computation of their ratio.

Unlike additive masking, where concrete bounds are known to achieve $\lambda$-bit security \cite{schoenmakers2006efficient}, the literature on multiplicative masking is limited and does not offer an established standard for achieving such guarantees. In our setting, we assume \(c\) is drawn from a distribution with sufficiently wide support such that the mapping from candidate \((\textit{num},\;\textit{denom})\) pairs to the observed output \((\textit{num} \cdot c,\;\textit{denom} \cdot c)\) remains indistinguishable across all candidates yielding the same AUC. As such, our approach focuses on obfuscation of the numerator and denominator values, rather than achieving indistinguishability against computationally bounded adversaries with respect to a specific bit-security level.

It is important to note that this masking strategy is only relevant in the context of semi-honest clients, particularly in edge cases such as when there are only two clients contributing to the secure computation. In contrast, the security of our method against the aggregator is maintained through FHE, and our CKKS parameters ensure 128-bit security under the standard cryptographic assumptions. Thus, the masking mechanism serves as an additional layer of privacy against semi-honest input parties within our framework and does not affect the security of the protocol against the aggregator.

\subsection{Security Analysis of the Semi-honest Setting}

We examine the information that can be obtained by a semi-honest aggregator as well as by semi-honest input parties individually.

\subsubsection{Semi-honest Aggregator}  
The local prediction confidence values, their distributions, label information, and the number of samples for each client are never shared and, therefore, cannot be known by a semi-honest aggregator due to the nature of the algorithm. The values $\text{T}^{\sdcp}$ and $\text{F}^{\sdcp}$, alongside $\text{T}^{\sizedcp - 1}$ and $\text{F}^{\sizedcp - 1}$ for $\sdcp \in \{0,\ldots,(\sizedcp-1)\}$, are the only information that each participant shares with the aggregator. Since all of these statistics are encrypted, a semi-honest aggregator cannot obtain the plaintext values, as ensured by Assumption 1. Similarly, the \textit{num} and \textit{denom} values computed by the aggregator at the end of the algorithm remain encrypted, making it impossible for the aggregator to infer the global model's performance, again based on Assumption 1.

Hence, the only way a semi-honest aggregator could break privacy is by obtaining the FHE private key, which would allow decryption of the shared ciphertexts. However, this is not possible under our assumptions. Assumption 2 ensures that communication between participants and the aggregator is secure, preventing any eavesdropping. Assumption 3 guarantees that the private key is distributed securely, making it inaccessible to the aggregator. Lastly, Assumption 4 prevents collusion between the aggregator and input parties, ensuring that the aggregator cannot gain access to the private key through cooperation.

\subsubsection{Semi-honest Input Parties}  
The only possible source of information leakage arises from the outsourced division operation performed by the aggregator. We illustrate that allowing input parties to obtain the \textit{num} and \textit{denom} values without masking would introduce a potential privacy risk.

Consider an edge case with two input parties where the \textit{denom} is the product of two prime numbers. In this case, an input party could deduce the number of samples held by the other party, as well as its class distribution. Since the \textit{denom} of AUC is the product of the total number of positive and negative samples, a \textit{denom} consisting of two prime factors has only one valid pair of values that satisfy this equation. If an input party holds a number of positive or negative samples greater than either of the prime factors, it can precisely determine which factor corresponds to the total number of positive samples and which corresponds to the total number of negative samples. Consequently, the input party could infer the total number of samples and the class distribution of the other party.  

To prevent this, our approach applies a random multiplicative mask before outsourcing the division operation. As we show in Section~\ref{sec:masking_proof}, this masking strategy ensures that the blinded outputs \textit{num}' and \textit{denom}' preserve the privacy of the underlying values in practice, as they reveal only the AUC ratio. While the masking does not provide formal $\lambda$-bit security guarantees, it effectively obscures the individual values of \textit{num} and \textit{denom} in typical semi-honest settings.

\subsection{Security Analysis of the Malicious Setting}

This section provides a security analysis for our federated AUC computation protocol designed to resist a malicious aggregator. Specifically, our protocol leverages:
\begin{itemize}
  \item \textbf{Zero-sum masks} random offsets across input parties that sum to zero. For each run we use two diffrent zero-sum masks,

  Each count (e.g.\ $\text{TP}_\sip^\delta$ and $\text{FP}_\sip^\delta$) is \emph{masked} by adding random offsets $R_{\sip}^\delta$ such that
\[
  \sum_{\sip=0}^{\sizeip-1} R_\sip^\delta = 0.
\]
Hence, the aggregator sees ciphertexts of $\text{TP}_\sip^\delta + R_{\sip}^\delta$, but cannot recover the true $\text{TP}_\sip^\delta$ unless it also knows all $R_{\sip}^\delta$.  Similarly for $\text{FP}_\sip^\delta$.
  
  \item \textbf{Random multipliers} $(r_0, r_1, r_2)$ (and $(r_0', r_1', r_2')$ for a second run) to \emph{blind} the final \(\textit{num}\) and \(\textit{denom}\).

  After partial sums are computed (still masked), input parties multiply them by random values such that
\[
  r_0 = r_3 r_4,\quad
  r_1 = r_5 r_6,\quad
  r_2 = r_7 r_8.
\]
The aggregator thus only sees \emph{blinded} versions of $(\textit{num}, \textit{denom})$:
\[
  r_0 \cdot \textit{num} + r_1 \cdot \textit{denom},
  \quad
  r_2 \cdot \textit{denom}.
\]
  
  \item \textbf{Random permutation} \(\pi\) (and \(\pi'\) for the second run), shared among input parties but unknown to the aggregator, which hides the correspondence of ciphertexts to threshold indices.

  Each party then \emph{splits} its masked counts into additive shares and places those shares into vectors, which are subsequently \emph{permuted} by a common random permutation $\pi$ (unknown to the aggregator).  As a result, the aggregator cannot identify which ciphertext belongs to which threshold $\delta$, nor which input party contributed which share.

  \item \textbf{Double Execution + Verification}
The entire protocol is then repeated using a \emph{second} independent set $(r_0',r_1',r_2')$.  If, after decrypting, the two resulting (blinded) ratios
\[
  \frac{r_0 \cdot \textit{num} + r_1 \cdot \textit{denom}}{r_2 \cdot \textit{denom}}
  \quad\text{and}\quad
  \frac{r_0' \cdot \textit{num} + r_1' \cdot \textit{denom}}{r_2' \cdot \textit{denom}}
\]
do not match, the aggregator is immediately caught.  If they match, with high probability the aggregator followed the protocol honestly.
\end{itemize}
Our main claim is that the aggregator cannot bias or tamper with the final AUC result without being caught with overwhelming probability.

\subsubsection{Security Model}

\paragraph{Malicious Aggregator.}
The aggregator $\mathcal{A}$ can arbitrarily deviate from the specified steps (e.g.\ drop or reorder ciphertexts, inject its own ciphertexts, or modify homomorphic additions).  However, it does \emph{not} know:
\begin{itemize}
  \item The \emph{private key} for the FHE encryption scheme,
  \item The zero-sum mask offsets $R_\sip^\delta$,
  \item The random permutations $\pi$ and $\pi'$,
  \item The random multipliers $\{r_0, r_1, r_2\}$ and $\{r_0', r_1', r_2'\}$ used in the two executions.
\end{itemize}

\paragraph{Goal.}
We show that $\mathcal{A}$ cannot produce a final AUC value different from $\frac{\textit{num}}{\textit{denom}}$ \emph{without} causing a detectable mismatch in the final verification.

\textbf{Correctness:} A scheme is correct if any honest computation will decrypt to the expected result \cite{viand2023verifiable}.

We demonstrate that conducting the homomorphic computation on the encrypted randomized values, following the protocol description, consistently yields the expected results: $E_{K}(r_0 \cdot \textit{num} + r_1 \cdot \textit{denom})$ and $E_{K}(r_2 \cdot \textit{denom})$.

\textbf{Soundness:} A scheme is sound if the adversary cannot make the verification accept an incorrect answer \cite{viand2023verifiable}.

\begin{theorem}
\label{thm:security-malicious}
Let $\lambda$ be the security parameter.  Under a semantically secure FHE scheme and assuming all random values (the zero-sum offsets, multipliers, and permutation) are chosen from sufficiently large domains, the probability that a malicious aggregator $\mathcal{A}$ can \emph{cheat} (i.e.\ alter the true AUC) yet \emph{still} pass the double-execution verification is negligible in $\lambda$.
\end{theorem}

\begin{proof}[Proof Sketch]

We break down the argument into four key steps:

\paragraph{1. FHE Security and Unknown Randomizers.}
By \emph{IND-CPA (or stronger) security} of the homomorphic encryption, the aggregator $\mathcal{A}$ learns no information about plaintext values beyond what is revealed by final decryptions.  In particular, $\mathcal{A}$ does not learn:
\begin{itemize}
  \item The zero-sum offsets $R_\sip^\delta$,
  \item The random permutation $\pi$ (and $\pi'$) (which scrambles the order of threshold shares),
  \item The random multipliers $r_0, r_1, r_2$ (and $r_0', r_1', r_2'$).
\end{itemize}
These unknowns \emph{blind} the aggregator’s view of the partial sums for $\textit{num}$ and $\textit{denom}$, preventing $\mathcal{A}$ from “surgically” manipulating them to a desired value.

\paragraph{2. Zero-Sum Masks + Permutation Prevent Targeted Attacks.}
Each threshold count $\text{TP}_\sip^\delta$ (or $\text{FP}_\sip^\delta$) is merged with random offsets $R_\sip^\delta$ so that $\sum_\sip R_\sip^\delta=0$.  Moreover, each party splits these masked counts into shares and permutes them using $\pi$.  Thus:
\begin{enumerate}
  \item \emph{No Threshold Index.} $\mathcal{A}$ cannot identify which ciphertext belongs to which $\delta$ (due to $\pi$).
  \item \emph{No Party-Specific Mask.} Even if $\mathcal{A}$ guessed a ciphertext’s threshold, the local zero-sum offset from other parties is unknown, so $\mathcal{A}$ cannot isolate or adjust a single party’s counts without creating a mismatch in the global sum.
\end{enumerate}
Hence, any attempt to reorder, omit, or alter these ciphertexts breaks the aggregator’s ability to produce consistent final sums for the same thresholds across all parties. The aggregator's ability to accurately determine the order of these ciphertexts in both secure computations is constrained by a probability of $(\frac{S!(SN-S)!}{SN!})^2$ where $N$ is the number of thresholds and $S$ is the split count (the number of additive shares). The split count \(S\) is crucial for our malicious setting's security. It directly affects the probability that a malicious aggregator can alter inputs, intermediate values, or outputs such that both computations yield the same modified AUC. Increasing \(S\) while decreasing the number of decision points \(N\) lowers the aggregator's success probability. For example, with \(N=100\) and \(S=7\), the probability is \(\frac{1}{2^{107}}\); with \(N=100\) and \(S=4\), it is \(\frac{1}{2^{60}}\); and with \(N=25\) and \(S=9\), it is \(\frac{1}{2^{103}}\).

\paragraph{3. Double Execution Requires Consistency Across Independent Randoms.}
Even if $\mathcal{A}$ tried to cheat in the \emph{first} run by tampering with the homomorphic sums (aiming to produce some incorrect $\widetilde{\textit{num}}$ or $\widetilde{\textit{denom}}$), it must also replicate the \emph{same} manipulation in the \emph{second} run so that:
\[
  \frac{r_0 \cdot \widetilde{\textit{num}}
        + r_1 \cdot \widetilde{\textit{denom}}}
       {r_2 \cdot \widetilde{\textit{denom}}}
  \;\;=\;\;
  \frac{r_0' \cdot \widetilde{\textit{num}}
        + r_1' \cdot \widetilde{\textit{denom}}}
       {r_2' \cdot \widetilde{\textit{denom}}}.
\]
Since $(r_0, r_1, r_2)$ and $(r_0', r_1', r_2')$ are \emph{independent, secret} random values, $\mathcal{A}$ effectively has to \emph{guess} how to align these tampered plaintexts to yield the \emph{same masked ratio} in both runs.  The probability of success without knowledge of the random multipliers is negligible in $\lambda$. The computational infeasibility of this verification is compounded by the probability of correctly guessing only six 32-bit random numbers in two secure computations, which stands at $\frac{1}{2^{192}}$ a statistically negligible probability and computationally impractical to achieve.

\paragraph{4. Final Verification Detects Deviations.}
After the aggregator’s two runs, the input parties decrypt to obtain:
\begin{align*}
    \text{AUC} 
  = \frac{r_0 \cdot \textit{num} + r_1 \cdot \textit{denom}}{r_2 \cdot \textit{denom}}, \\
  \text{AUC}'
  = \frac{r_0' \cdot \textit{num} + r_1' \cdot \textit{denom}}{r_2' \cdot \textit{denom}}.
\end{align*}
  
They compare $\text{AUC}$ and $\text{AUC}'$. Any inconsistency caused by aggregator cheating is caught when $\text{AUC} \neq \text{AUC}'$.  If $\text{AUC}=\text{AUC}'$, the aggregator must have computed faithfully with overwhelming probability. Because of the combined effects of FHE security, zero-sum masking, unknown permutation, random multipliers, and a two-run consistency check, the aggregator cannot tamper with the final AUC without producing a detectable inconsistency in at least one run.  Hence, the probability of cheating and passing verification is negligible. This completes the proof of Theorem~\ref{thm:security-malicious}.
\end{proof}

\textbf{Security:} 
Our verification method in the malicious setting primarily relies on homomorphic decryption. To compromise the security of the verification method, the adversary would need to break the security of the underlying FHE scheme, Cheon-Kim-Kim-Song (CKKS), the security of which has been formally proven in \cite{cheon2017homomorphic}.

\subsection{Equality Check Through CKKS}

The CKKS scheme \cite{cheon2017homomorphic} introduces noise into computations, which accumulates with each operation. As a result, exact equality between two AUC scores computed in the malicious setting is not feasible in theory. However, based on our experimental results, we observe that in a practical setting, despite the noise introduced by CKKS, the maliciously secure AUC scores remain comparable up to five decimal places. Therefore, our equality check verifies the scores up to five digits of precision.

\end{document}